\documentclass[reqno]{amsart}

\usepackage{geometry,amsmath,amsfonts,amssymb,amsthm,amscd,mathrsfs,graphicx,enumerate,multicol,wrapfig,subfigure,hyperref, slashed}
\usepackage[all]{xy}
\usepackage{ upgreek }
\usepackage[toc]{appendix}

\usepackage[T1]{fontenc}

\numberwithin{equation}{subsection} 

\newcommand{\ra}{\rightarrow}

\newcommand{\lra}{\longrightarrow}
\newcommand{\8}{\infty}
\newcommand{\p}{\prime}
\newcommand{\pt}{\partial}

\newcommand{\e}{\epsilon}
\newcommand{\al}{\alpha}
\newcommand{\Om}{\Omega}

\newcommand{\gam}{\gamma}
\newcommand{\Gam}{\Gamma}

\newcommand{\vp}{\varphi}
\newcommand{\lam}{\lambda}

\newcommand{\q}{\theta}

\newcommand{\dt}{\delta}

\newcommand{\Zbb}{\mathbb{Z}}
\newcommand{\Pbb}{\mathbb{P}}
\newcommand{\Cbb}{\mathbb{C}}

\theoremstyle{plain} 

\newtheorem{THM}{Theorem}[section]
\newtheorem{DEF}[THM]{Definition}

\newtheorem{PROP}[THM]{Proposition}
\newtheorem{LEM}[THM]{Lemma}
\newtheorem{COR}[THM]{Corollary}
\newtheorem{REM}[THM]{Remark}

\newcommand{\drm}{\mathrm{d}}
\newcommand{\bt}{\bullet}

\newcommand{\Vc}{\mathcal{V}}
\newcommand{\Uc}{\mathcal{U}}
\newcommand{\Wc}{\mathcal{W}}
\newcommand{\Oc}{\mathcal{O}}
\newcommand{\red}{\mathrm{red}}

\newcommand{\Jc}{\mathcal{J}}

\newcommand{\Cc}{\mathcal{C}}
\newcommand{\Gc}{\mathcal{G}}

\newcommand{\Dc}{\mathcal{D}}

\newcommand{\Fc}{\mathcal{F}}
\newcommand{\Lc}{\mathcal{L}}

\newcommand{\Xfr}{\mathfrak{X}}
\newcommand{\Ec}{\mathcal{E}}

\newcommand{\Ufr}{\mathfrak{U}}
\newcommand{\Ber}{\mathrm{Ber}}

\newcommand{\Mcl}{\mathcal{M}}

\title{
Superconformal structures and the Superparticle on $\Cbb\Pbb^1$
\\}
\author{\small Kowshik Bettadapura
}
\date{}

\begin{document}
\maketitle


\begin{abstract}
\noindent 
In this paper the notion of a \emph{superconformal structure} on a supermanifold is introduced in an effort to study the  superparticle sigma-model. There are, in particular, two main aspects of the sigma-model which are investigated. The first is on the relationship between the superparticle Lagrangian and the component Lagrangian; and the second is on the problem of integrating infinitesimal variations to globally-defined variations of the component Lagrangian, which leads naturally to a notion of \emph{consistency}. Throughout this paper illustrations are provided on the complex projective line. 
\end{abstract}

\tableofcontents

\section{Introduction}

One of the main themes motivating this paper is on understanding the interplay between \emph{supergeometry} and \emph{supersymmetric field theory}. Typically, in what is referred to as the \emph{component formalism} \cite{FREEDSUSY}, one formulates a supersymmetric field theory on a supermanifold and subsequently ``integrates out'' the fermionic variables, so as to obtain quantities defined on the underlying manifold. Then the techniques of differential geometry, such as methods of variational calculus, become applicable. For instance, given a Lagrangian density $\Lc$ defined on a supermanifold $\Xfr$, if the supermanifold is nice enough (in a certain sense) one can integrate $\Lc$ to get the \emph{component Lagrangian} $L$, which is then a certain density defined on the underlying manifold $M$. One may then take variational derivatives of $L$ and optimise it to obtain equations of motion\footnote{Of course, the equations of motion themselves may be derived by other means. In this paper however we are only interested here in their derivation via the component formalism.}.

The caveat, glossed over above, in order for the component formalism to yield meaningful quantities, is for the supermanifold be ``nice enough''. More precisely, it is meant that the supermanifold in question be \emph{split}. Roughly, one of the key morals of this note is that even this caveat is not strong enough, thereby illustrating some of the subtleties associated with the geometry of supermanifolds. As an instructive case study, the superparticle sigma-model on a supermanifold over $\Cbb\Pbb^1$ is investigated, and it is for this reason that the notion of a \emph{superconformal structure} is introduced. 
\\

The superparticle Lagrangian is a density defined by a certain collection of vector fields, called \emph{superconformal vector fields} (\cite{RAB, WNSRS}), which are themselves subject to a certain algebraic constraint---that they satisfy the relations of a \emph{supersymmetry algebra} (for its relation to physics see \cite{WITTSUSY}). With these vector fields we are led to a notion of \emph{maximal superconformal structure}. If a supermanifold admits such a structure, then many questions and properties  regarding the superparticle Lagrangian $\Lc$ become tractable, so a natural question is then: \emph{when does a supermanifold admit such a structure?} As will be argued, that such a structure be admitted will be quite rare (at least in the $(1|2)$-dimensional case) and indeed a large class of supermanifolds (e.g., the non-split ones) can be seen to fail to admit such structures. The set of all $(1|2)$-dimensional supermanifolds over $\Cbb\Pbb^1$ which admit a maximally superconformal structure is described in this paper. 

Regarding the tractable questions and properties about $\Lc$ alluded to above, we focus on two in particular: firstly, on the relationship between $\Lc$ and the component Lagrangian $L$. And secondly, on the notion of infinitesimal variations of $L$ and the problem of integrability thereof. It will be shown that even in the presence of a maximal superconformal structure, the relationship between $\Lc$ and $L$ can be quite complicated and perhaps even a little nebulous. As for the problem of integrability, this leads to a notion of \emph{consistency} for $L$. The motivating question behind this notion being: \emph{is it possible to ``consistently'' derive the equations of motion, via the component formalism, in the given supersymmetric field theory?} To this we answer: \emph{yes}, if the infinitesimal variations integrate to give, in a sense, a globally well-defined variation. Finally, it is observed that $L$ will in general fail to be consistent. However, it will be possible to identify a locus of \emph{consistent fields for $L$}. In applications on $\Cbb\Pbb^1$ we endeavour to illustrate the failure for $L$ to be consistent and subsequently describe this locus. 
\\

This paper is organised as follows. In Section 2 some preliminary theory on supermanifolds is provided and the notion of a (maximal) superconformal structure is introduced. The conditions for a supermanifold to admit such a structure is identified and 
illustrations are given on $\Cbb\Pbb^1$. The relationship between the superparticle Lagrangian $\Lc$ and the component Lagrangian $L$ is elaborated on, and the notion of consistency is introduced in Section 3. Finally, in Section 4, the considerations of Section 3 are applied to $\Cbb\Pbb^1$, where the locus of consistent fields for the component Lagrangian $L$ is described.

\subsection*{Acknowledgements} This research was undertaken at the Mathematical Sciences Institute at the Australian National University and supported by an Australian Postgraduate Award. I would like to thank Peter Bouwknegt for suggesting the notion of consistency and for subsequent helpful discussions and useful comments made on this paper.

\section{Preliminary Theory}
\label{rjncrnckmrkcr}

In this section we provide some background theory on supermanifolds and introduce the notion of a superconformal structure. We conclude with illustrations on $\Cbb\Pbb^1$.

\subsection{Supermanifolds}
The definition of a supermanifold may be quite succinctly given in the framework of algebraic geometry. We refer to \cite{BZ, QFASDM} where this point of view is emphasised. 	

\begin{DEF}
\emph{
A $(p|q)$-dimensional real (resp. complex) supermanifold $\Xfr$ is defined as a locally ringed space $(M, \Oc_M)$, where $M$ is a $p$-dimensional, real (resp. complex) manifold and the structure sheaf $\Oc_M$ is locally isomorphic to the sheaf of rings $\Cc_M\otimes \wedge^\bt(V^q)$, where $V$ is a real (resp. complex) vector space and $\Cc_M$ is the ring of functions on $M$.
}
\end{DEF}

For $\Xfr$ of dimension $(p|q)$, we say it has \emph{even} dimension $p$ and \emph{odd} dimension $q$. For the purposes of this paper it will be convenient to think about a supermanifold $\Xfr$ as being modelled on two bits of data: a manifold $M$, called the \emph{reduced space}, and a vector bundle $E\ra M$. Then the structure sheaf $\Oc_M$ of $\Xfr = (M, \Oc_M)$ will be locally isomorphic to the sheaf of sections of the bundle of exterior algebras $\wedge^\bt \Ec$, where $\Ec= \Gam(M, E)$. In this paper we will consider $\Xfr_{(M, E)}$, for $E$ is a holomorphic vector bundle, a \emph{complex supermanifold}. This is in contrast to those studied in \cite{WNS} where only the underlying manifold $M$ is assumed to admit any complex structure.  

It is useful to note that there exists a canonical inclusion $M\hookrightarrow \Xfr$ as locally ringed spaces, corresponding to the ``body-map'' $\e^*: \Oc_M \ra \Oc_M/\Jc$, where $\Jc\subset \Oc_M$ denotes the subsheaf generated by nilpotent elements. Now from any given pair $(M, E)$ we may construct a supermanifold by simply taking the structure sheaf $\Oc_M$ to be $\wedge^\bt \Ec$. This gives what is termed the \emph{split model} (for instance, in \cite{WD}) and is denoted $\Pi E$. As a locally ringed space $\Pi E =  (M, \wedge^\bt\Ec)$. A supermanifold $\Xfr = (M, \Oc_M)$ equipped with a choice of local isomorphisms $\Oc_M \cong_{\mathrm{loc}}\wedge^\bt \Ec$ is then said to be \emph{modelled} on $(M, E)$. If we wish to only specify the underlying manifold $M$, then we will say $\Xfr$ is a supermanifold \emph{over $M$}. As shorthand we write $\Xfr_{(M, E)}$ to mean $\Xfr$ is a supermanifold modelled on the manifold $M$ and the vector bundle $E\ra M$. That every supermanifold may be written in this way for some $(M, E)$, up to isomorphism, is described in \cite[Proposition 2, p. 588]{GR}.

\begin{DEF}\label{rjcnrncjkrccr}
\emph{
The supermanifold $\Xfr_{(M, E)} = (M, \Oc_M)$ is said to be \emph{split} if the sheaves $\Oc_M$ and $\wedge^\bt\Ec$ are isomorphic as $\Cc_M$-modules.
}
\end{DEF}

The supermanifold $\Xfr$ is said to be \emph{non-split} if it is not split. The above definition is meaningful only in the holomorphic category. Indeed, in the smooth category we have the following structure theorem for supermanifolds, first proved by Batchelor \cite{BAT}.

\begin{THM}
As a smooth supermanifold, $\Xfr_{(M, E)}$ is split.\qed
\end{THM}

That a supermanifold be split is a highly desirable property from the perspective of a physicist interested in studying supersymmetry. This is because, in this case, one can formulate quantities of interest on the supermanifold itself, such as Lagrangians, and integrate them by a process of what is known as \emph{Berezin integration}. The Berezinian line bundle on a supermanifold plays a role analogous to the canonical line bundle on a complex manifold.  For the purposes of this paper, we need only be aware that it is sections of this line bundle on which computes a Berezin integral. Indeed, on a split supermanifold, we have the following nice characterisation, which was observed in the case of \emph{supercurves} \cite[p. 7]{BERGRAB}, but holds in more generality. 

\begin{LEM}\label{jnckncjrkmcrkc}
The sheaf of sections of the Berezinian line bundle on a split supermanifold $\Xfr_{(M, E)}$ is isomorphic, as modules over $\Cc_M$, to $\Om^{\mathrm{top}}(M)\otimes \det \Ec^\vee$. \qed
\end{LEM}

As will be observed, the above lemma has a nice consequence for what we term \emph{maximally superconformal supermanifolds}, to be introduced and discussed in the section to come. Throughout the rest of this paper we will only be concerned with $(1|2)$-dimensional supermanifolds and will prefer to work with open covers and trivialisations. That the definitions made, and results obtained, are independent of these choices is in principle true but not explicitly stated or proved.

\subsection{Superconformal Structures} Let $\Xfr$ be a $(1|2)$-dimensional, complex supermanifold with covering $\mathfrak U = \{\Uc, \Vc, \ldots\}$ and transition functions $\rho = (\rho_{\Uc\Vc})$. We will often refer to the pair $(\Ufr, \rho)$ as a trivialisation of $\Xfr$. Here $\Uc\subset \Xfr$ but, by abuse of notation, we will also identify it with an open subset of $\Cbb^{1|2}$. Underlying the covering $\Ufr$ for $\Xfr$ is the covering $\Ufr_\red = \{U, V, \ldots\}$ of the underlying manifold $M$, and $\Uc$ can be obtained from $U$ in the sense that we have $\Uc = U\times \Cbb^{0|2}$. Following the formalism in \cite{BER}, let  $(X, \q^1, \q^2)$ denote a system of coordinates on $\Uc$ (with $x = X_\red$ being the coordinate on $U$). Then on $\Uc$ we may write down the following vector fields,
\begin{align}
D_{\Uc, 1} = \frac{\pt}{\pt \q^1} - \q^2\frac{\pt}{\pt X}
&&
\mbox{and}
&&
D_{\Uc, 2} = \frac{\pt}{\pt \q^2} - \q^1\frac{\pt}{\pt X}.
\label{jcnkrnckrnkcr}
\end{align}
If $(Y, \eta^1,\eta^2)$ denotes a system of coordinates on $\Vc$ then we similarly have the vector fields $D_{\Vc, a}$, for $a = 1, 2$. Let $\Dc_a = \{D_{\Uc, a}, D_{\Vc, a}, \ldots\}$ be a collection of such vector fields defined locally. The interest in studying such vector fields stems from the observation that their (super-)Lie bracket, taken in the tangent bundle of $\Xfr$ which is viewed as a (super-)Lie algebra, is given by
\[
\frac{1}{2}[\Dc_{\Uc, 1}, \Dc_{\Uc, 2}] = -\frac{\pt}{\pt X}.
\]  
Upon quantisation, in a certain sense, the above vector fields would be replaced by operators and the above relation is precisely that which defines the supersymmetry algebra (see \cite{FREEDSUSY, WITTSUSY}). In studying such vector fields on supermanifolds, one could argue that we are investigating geometric properties of the supersymmetry algebra. 

\begin{DEF}\emph{
A vector field of the form given in \eqref{jcnkrnckrnkcr} is termed a \emph{superconformal vector field}. 
}
\end{DEF}

\begin{REM}
\label{jdckcnkcjkrmckjr}
\emph{
We refer to \cite{CRARAB, RAB, WNSRS} and references therein for usage of the term ``superconformal'' with regards to these vector fields. 
}
\end{REM}

Note that choices of coordinate system were used to order to define the superconformal vector fields. It is natural to ask whether these locally defined vector fields patch together to yield a globally defined vector field, i.e., a vector field on the supermanifold itself. In more detail, let $\Xfr_{(M, E)}$ come equipped with a trivialisation $(\Ufr,\rho)$. That the superconformal vector fields be globally defined means there exist a 1-cocycle of matrix-valued, holomorphic functions $g(\Ufr) = (g^b_{\Uc\Vc, a})$, for $a,b = 1, 2$, such that
\begin{align}
(\rho_{\Uc\Vc})_*D_{\Uc, a} = g^b_{\Uc\Vc, a} D_{\Vc, b}.
\label{jcnknckncjkck}
\end{align}
Here $(\rho_{\Uc\Vc})_*$ is the natural pushforward map on vector fields induced by the trivialisation $\rho$, and the Einstein summation convention is being used. 

\begin{DEF}
\emph{
The supermanifold $\Xfr_{(M, E)}$ admits a \emph{superconformal structure} if there exists a trivialisation of $\Xfr_{(M, E)}$ in which the superconformal vector fields are globally defined. 
}
\end{DEF}

The cocycle $g(\Ufr) = \{(g^b_{\Uc\Vc,a})\}$ is said here to be the \emph{defining cocycle for the superconformal structure}. 
Now \eqref{jcnknckncjkck} does \emph{not} imply that $\Dc_a$, $a = 1, 2$ are, individually, globally defined vector fields on $\Xfr$. In order for this to hold we would have to require the defining cocycle $g(\Ufr)$ be diagonal, i.e., comprise of a collection of diagonal matrices. This leads us to the notion of maximality.

\begin{DEF}\label{rjcbjkrcbnjkr}
\emph{
The supermanifold $\Xfr_{(M, E)}$ admits a \emph{maximal} superconformal structure if it admits a superconformal structure whose defining cocycle is diagonal. 
}
\end{DEF}

%
%
%
%
%
%
%
%
%
%
 
%
\begin{REM}
\emph{
The notion of superconformal structure so far introduced is only applicable in the $(1|2)$-dimensional case. It is possible however to generalise this notion to supermanifolds of dimension other than $(1|2)$. In the $(1|1)$-dimensional case, the appropriate generalisation will lead to what is commonly known as a \emph{super-Riemann surface}. We refer to a forthcoming work \cite{BETT} where superconformal structures in higher-dimensional cases are studied. 
}
\end{REM}

The question of whether the given, locally defined, superconformal vector fields are globally defined can now be translated into the question of whether the supermanifold itself admits a superconformal structure. In the case of maximal superconformal structures we have the following result.

\begin{THM}\label{chncnrcnrk}
The $(1|2)$-dimensional supermanifold $\Xfr_{(M, E)}$ is maximally superconformal if and only if it is the split model $\Pi E$, where $E$ satisfies:
\begin{enumerate} [$(1)$]
	\item it splits into a direct sum of line bundles and;
	\item its determinant line bundle $\det E$ is identified with the tangent bundle $TM$.
\end{enumerate}
\end{THM}

\begin{proof}
We defer a proof of this theorem to Appendix \ref{ejnckenckrckrrc}. 
\end{proof}

A useful corollary of part (2) in the above theorem and Lemma \ref{jnckncjrkmcrkc} is the corresponding statement about the Berezinian line bundle.

\begin{COR}\label{rcjnkrnckrnckr}
The Berezinian Line bundle of a $(1|2)$-dimensional, maximally superconformal supermanifold is trivial.\qed
\end{COR}

%

\subsection{Illustrations on $\Cbb\Pbb^1$}
The applications in this paper  revolve around $\Cbb\Pbb^1$, so it is therefore pertinent to include a brief discussion of superconformal structures on supermanifolds over $\Cbb\Pbb^1$. Such supermanifolds have been extensively studied by many authors and the notation used here may be attributed to that used in \cite{BONIS, VISH}.
\\

We start off here with a theorem, first proved by Grothendieck and stated and proved in \cite[p. 12]{OSS}, on the nature of holomorphic vector bundles on $\Cbb\Pbb^1$.

\begin{THM}\label{rjcnkrnckrcnkr}
The sheaf of sections $\Ec$ of a holomorphic vector bundle $E$ over $\Cbb\Pbb^1$ of rank $r$ may be written as 
\[
\Ec = \Oc_{\Cbb\Pbb^1}(k_1) \oplus \cdots \oplus \Oc_{\Cbb\Pbb^1}(k_r)
\]
for unique $k_i\in \Zbb$ up to re-ordering. \qed
\end{THM}

Hence we need only two integers $(-k_1, -k_2)$, representing degrees of line bundles, to describe any rank-two vector bundle on $\Cbb\Pbb^1$. As such we will simply write $\Xfr_{(-k_1, -k_2)}$ for a $(1|2)$-dimensional supermanifold $\Xfr$ modelled on $\Cbb\Pbb^1$ and the vector bundle $E \equiv (-k_1, -k_2)$. In order to construct an explicit trivialisation for $\Xfr_{(-k_1,- k_2)}$ we firstly need one for $\Cbb\Pbb^1$, and here we employ the usual trivialisation, described in detail in \cite{GH}. This trivialisation consists of two open sets $U$ and $V$. If $x$ denotes the complex coordinate on $U$ and $y$ that on $V$, then they are related by the transition functions $f = \{f_{UV}\}$ as follows,
\[
y \sim f_{UV}(x) = \frac{1}{x}.
\]
We refer to this trivialisation of $\Cbb\Pbb^1$ as the \emph{standard trivialisation}. The supermanifold $\Xfr_{(-k_1, -k_2)}$ is now covered by $\Uc$ and $\Vc$ where $U = \Uc_\red$ and $V = \Vc_\red$. Let $(X, \q^1, \q^2)$ denote coordinates on $\Uc$ and $(Y, \eta^1, \eta^2)$ coordinates on $\Vc$. We may identify the coordinates $x = X_\red$ and $y = Y_\red$ with those on the reduced space $\Cbb\Pbb^1$. The transition functions for $\Xfr_{(-k_1, -k_2)}$, denoted $\rho = (\rho_{\Uc\Vc})$, are given by
\begin{align}
\left(
\begin{array}{l}
Y\\
\eta^1\\
\eta^2
\end{array}
\right)
\sim 
\rho_{\Uc\Vc}(X, \q^1, \q^2)
=
\left(
\begin{array}{c}
x^{-1} + \al_{UV}(x)~\q^1\wedge\q^2\\
\lam_1x^{-k_1}\q^1\\
\lam_2x^{-k_2}\q^2
\end{array}
\right),
\label{ejxnnxjkenxjken}
\end{align}
where $\lam_1, \lam_2$ are non-zero, complex constants, and $\al_{UV}$ is a holomorphic function on $U\cap V \cong \Cbb^\times$. If $\Xfr_{(-k_1, -k_2)}$ is taken to be the split model then we must set $\{\al_{UV}\} = 0$ in \eqref{ejxnnxjkenxjken}. \\

Now let $\lam = \lam_1\oplus \lam_2$ be the two-by-two matrix with $\lam_1$ and $\lam_2$ on the diagonal. For such a matrix we denote by $\Pi_{(-k_1, -k_2)}(\lam)$ the \emph{split model} given by the trivialisation in \eqref{ejxnnxjkenxjken}. Here $\Xfr_{(-k_1, -k_2)}$ need not be isomorphic to $\Pi_{(-k_1, -k_2)}(\lam)$. However both $\Pi_{(-k_1, -k_2)}(\lam)$ and $\Pi_{(-k_1, -k_2)}(\lam^\p)$ are isomorphic for any two invertible, diagonal matrices $\lam, \lam^\p$ and typically one would simply write $\Pi_{(-k_1, -k_2)}$, should a trivialisation not be specified. The reason we have introduced $\Pi_{(-k_1, -k_2)}(\lam)$ is to show that there exist trivialisations of $\Pi_{(-k_1, -k_2)}$ which are \emph{not} superconformal. Indeed, as can be deduced from Theorem \ref{chncnrcnrk}, we have the following:

\begin{COR}
\label{rjcnkrnjckrncjkr}
The split model $\Pi_{(-k_1, -k_2)}(\lam)$ is maximally superconformal iff $k_1 + k_2 = 2$ and $\det\lam = -1$.
\qed
\end{COR}

In accordance with Definition \ref{rjcbjkrcbnjkr}, we see from the above corollary and the preceding discussion that $\Pi_{(-k_1, -k_2)}$ will admit a maximally superconformal structure iff $k_1 + k_2 = 2$. In the literature there exists a construction of the supermanifold $\Cbb\Pbb^{m|n}$. It is obtained as the quotient space $(\Cbb^{m+1|n}- \{0\})/\Cbb^\times$ in much the same way as $\Cbb\Pbb^m = (\Cbb^{m+1} - \{0\})/\Cbb^\times$. This supermanifold is an instance of a more general notion of a \emph{supergrassmannian}, detailed in \cite{YMAN}. In the notation used in this paper, the supermanifold $\Cbb\Pbb^{1|2}$ is given, up to isomorphism, by $\Pi_{(-1, -1)}$. In particular, we see from Corollary \ref{rjcnkrnjckrncjkr} that it will admit a maximally superconformal structure. In light of these observations we set $\Pi_{(-1, -1)} = \Cbb\Pbb^{1|2}$ and term the supermanifold $\Pi_{(-k_1, -k_2)}$, for $(k_1, k_2)\neq(1, 1)$, a \emph{non-standard $\Cbb\Pbb^{1|2}$}. 


\section{The Superparticle Lagrangian and Consistency}
\label{rjnvknkjrmcr}

Let $\Xfr = (M, \Oc_M)$ be a supermanifold. We fix here the set of fields 
\begin{align}
\Fc = \left\{ \Phi : \Xfr \lra\Cbb^{1|1}\mid \mbox{$\Phi$ is smooth}
\right\}
\label{ecfcjhbhcjbdfc}
\end{align}
where, by requiring $\Phi$ be smooth it is meant that its component fields are smooth. As such we can identify $\Fc$ with the sheaf of \emph{smooth} functions on $\Xfr$. In this section we investigate the Lagrangian of the superparticle as a quantity defined on the supermanifold $\Xfr$. It is a functional which sends each field $\Phi$ to a density on $\Xfr$, i.e., a smooth section of the Berezinian line bundle $\Ber~\Xfr$. For an interpretation of the Lagrangian as a density defined on the space of fields, and valued in differential forms, see \cite[p. 26]{FREEDSUSY}. Some key issues addressed in this section are conditions under which the superparticle Lagrangian is globally well-defined, and on its relation to the \emph{component Lagrangian}.


\subsection{The Superparticle Lagrangian}\label{jckrnckrnckjrkxe}
Let $(\Ufr,\rho)$ denote a trivialisation of $\Xfr$. Wth respect to $\Ufr$, and for each field $\Phi$, we introduce the
\emph{superparticle Lagrangian zero-cochain}, or simply the \emph{superparticle Lagrangian}, as the following zero-cochain $\Lc(\Phi) = \{\Lc_\Uc(\Phi), \Lc_\Vc(\Phi), \ldots\}$ where
\[
\Lc_\Uc(\Phi) = \frac{1}{2}\e^{ab}\left\langle D_{\Uc,a}\Phi, D_{\Uc,b}\Phi\right\rangle~[\drm X~\drm\q^1 \drm\q^2]
\]
for $\e^{ab}$ antisymmetric, i.e., $\e^{12} = 1 = -\e^{21}$ and zero otherwise; and $[\drm X~\drm\q^1\drm\q^2]$ the trivialising section of the Berezinian over $\Uc$ (we refer to \cite{WNS} for the notation). For an analysis of the superparticle in the case where $\Xfr$ is affine superspace, see \cite[p. 42]{FREEDSUSY}. 
\\

With respect to the description of the superparticle Lagrangian given so far, we formulate the notion of it being ``global'' as follows. Let $\underline\Ber~\Xfr$ denote the sheaf of holomorphic sections of the Berezinian line bundle\footnote{ 
For the purposes of this paper it is irrelevant to provide details of Berezinians for supermanifolds in full generality. This is because, in applications to follow, the Berezinian will in fact be constant (see Corollary \ref{rcjnkrnckrnckr}).
} 
$\Ber~\Xfr$ and $\underline \Ber^\8\Xfr$ the sheaf of smooth sections. Then $\Lc(\Phi)$ lies in $C^0(\Ufr;\underline \Ber^\8\Xfr)$. There is a natural coboundary operator $\dt : C^0(\Ufr;\underline\Ber^\8\Xfr)\ra  C^1(\Ufr;\underline\Ber^\8\Xfr)$ sending $\Lc(\Phi) \ra (\dt\Lc)(\Phi)$, which is following the quantity defined on intersections:
\[
(\dt\Lc)_{\Uc\Vc}(\Phi) = \Lc_\Vc(\Phi) - \Lc_\Uc(\Phi).
\]
Hence if $(\dt\Lc)_{\Uc\Vc}(\Phi) = 0$ for all $\Uc, \Vc\in \Ufr$ with non-empty intersection, it follows that $\Lc(\Phi)$ is a global section of $\underline\Ber^\8\Xfr$. This is however too strong a requirement for our purposes in this paper. We consider instead the weaker case where $\Lc(\Phi)$ is a global section of $\underline\Ber^\8\Xfr\otimes \underline \ell^\8$, for  $\ell\ra \Xfr$ a line bundle. Here $\underline \ell$ denotes the sheaf of holomorphic sections of $\ell$ and $\underline\ell^\8$ that of smooth sections. 

\begin{DEF}
\label{fjcnkenckec}
\emph{
The superparticle Lagrangian $\Lc$ is said to be a \emph{global superparticle Lagrangian} for $\Xfr$ if there exists a line bundle $\ell\ra \Xfr$ and a smooth section $s$ of $\ell$ such that $\Lc(\Phi)\otimes s$ is a global section of  $\underline\Ber^\8\Xfr\otimes \underline \ell^\8$ for all fields $\Phi$ in $\Fc$. 
 }
 \end{DEF}

%
%
%
%
%

%

To elaborate on the above definition, let $(\Ufr,\rho)$ denote a trivialisation of $\Xfr$. On a non-empty intersection $\Uc\cap\Vc$ there is a natural way to compare $\Lc_\Uc(\Phi)$, a quantity defined on $\Uc$, with $\Lc_\Vc(\Phi)$, a quantity defined on $\Vc$, by means of the transition functions $\rho$. Then to require $\Lc$ be a \emph{global} superparticle Lagrangian in the sense of Definition \ref{fjcnkenckec} it is both necessary and sufficient for there to exist a multiplicative 1-cocycle of smooth functions $h(\Ufr) = \{h_{\Uc\Vc}\}$ such that 
\begin{align}
(\rho_{\Uc\Vc})_*\Lc_\Uc(\Phi) = h_{\Uc\Vc} \Lc_\Vc(\Phi) 
\label{jcnkrnckjrcnklrcmlr}
\end{align}
for all fields $\Phi$. The cocycle $h(\Ufr)$ is then said to be a \emph{defining cocycle} for $\Lc$. This allows one to identify sufficient conditions for the superparticle Lagrangian to be global, as in the following proposition.

\begin{PROP}\label{ejcnknckrncjrknr}
Let $\Xfr$ be a $(1|2)$-dimensional, maximally superconformal supermanifold. If the hermitian product\footnote{
For two complex variables $z$ and $w$, their \emph{hermitian product} is defined to be $z\overline w$.
} of the defining cocycles of its superconformal structure is real, then $\Lc$ will be a global superparticle Lagrangian.
\end{PROP}

\begin{proof}
Recall from Definition \ref{rjcbjkrcbnjkr} that if $\Xfr$ is maximally superconformal, then there exists a trivialisation $(\Ufr, \rho)$ in which it is superconformal and the defining cocycle $g(\Ufr)$ is diagonal. Set $g(\Ufr) = \{(g_{\Uc\Vc, a})\}$ where $g_{\Uc\Vc,a} := (g_{\Uc\Vc,a}^a)$. The left-hand side of \eqref{jcnkrnckjrcnklrcmlr} is computed as follows,
\begin{align}
 (\rho_{\Uc\Vc})_*\Lc_\Uc(\Phi)
&= \frac{1}{2}\e^{ab}\left\langle (\rho_{\Uc\Vc})_*D_{\Uc,a}\Phi, (\rho_{\Uc\Vc})_*D_{\Uc,b}\Phi\right\rangle~
\rho_{\Vc\Uc}^*[\drm X~\drm\q^1 \drm\q^2]
\notag
\\
 &= \frac{1}{2}\e^{ab} g_{\Uc\Vc, a}\overline{g_{\Uc\Vc, b}} 
\left\langle D_{\Vc,a}\Phi, D_{\Vc,b}\Phi\right\rangle
~\Ber~J(\rho_{\Vc\Uc})~[\drm Y~\drm\eta^1\drm\eta^2].
\notag\\
&= 
\frac{1}{2}\e^{ab} g_{\Uc\Vc, a}\overline{g_{\Uc\Vc, b}} 
\left\langle D_{\Vc,a}\Phi, D_{\Vc,b}\Phi\right\rangle
~[\drm Y~\drm\eta^1\drm\eta^2] 
\label{rchbrjcbhrjk}
\end{align}
where, in the last line above, we have used Corollary \ref{rcjnkrnckrnckr} to justify setting $\Ber~J(\rho_{\Vc\Uc}) = 1$.  Now set 
\begin{align}
\label{dkjcnnckfnckf}
G_{\Uc\Vc} := g_{\Uc\Vc, 1}\overline{g_{\Uc\Vc, 2}}.
\end{align}
Here $G_{\Uc\Vc}$ is a smooth, $\Cbb^{1|1}$-valued function on $\Uc\cap \Vc$ and we can make sense of writing it in terms of its real and imaginary parts component-wise. Denote these $\Re_{\Uc\Vc}$ and $\Im_{\Uc\Vc}$ respectively so that 
$
G_{\Uc\Vc} = \Re_{\Uc\Vc} + i\Im_{\Uc\Vc},
$
where $i = \sqrt{-1}$. Then from \eqref{rchbrjcbhrjk} we find,
\begin{align}
(\rho_{\Uc\Vc})_*\Lc_\Uc(\Phi) - \Re_{\Uc\Vc}\Lc_\Vc(\Phi) = i \Im_{\Uc\Vc} \mathcal M_\Vc(\Phi),
\label{jcnkncjknckjncke}
\end{align}
where we have set
\begin{align}
\mathcal M_\Vc(\Phi) := \frac{1}{2} s^{ab}\left\langle D_{\Vc,a}\Phi, D_{\Vc,b}\Phi\right\rangle
~[\drm Y~\drm\eta^1\drm\eta^2]
\label{jcnkrnckrmckjmrkjr}
\end{align}
for $s^{12} = s^{21} = 1$ and zero otherwise. The proposition now follows from \eqref{jcnkncjknckjncke}.
\end{proof}

For $\Xfr$ maximally superconformal, set $G(\Ufr) := \{G_{\Uc\Vc}\}$ where $G_{\Uc\Vc}$ is the Hermitian product of the defining cocycles introduced in \eqref{dkjcnnckfnckf}. In order to obtain both necessary and sufficient conditions for $\Lc$ to be global, we introduce the following definition. 

\begin{DEF}
\label{jncklfncknjckrcr}
\emph{
Suppose $\Lc$ is a global superparticle Lagrangian for a maximally superconformal supermanifold $\Xfr$. Then $\Lc$ is said to be \emph{compatible} with the superconformal structure if there exists a trivialisation $(\Ufr,\rho)$ in which $h(\Ufr) = G(\Ufr)$. 
}
\end{DEF}

 Now let $\Xfr$ be maximally superconformal. We see from Proposition \ref{ejcnknckrncjrknr} that the existence of a compatible, global superparticle Lagrangian is now \emph{equivalent} to requiring the product of the defining cocycles for the superconformal structure be real. In the sections to follow it will be shown that such compatible Lagrangians do \emph{not} exist over $\Cbb\Pbb^1$.

\subsection{The Component Lagrangian}\label{djckncjkrnckjr}
In supersymmetric field theories with a prescribed Lagrangian (defined on the supermanifold), one typically looks to reduce this Lagrangian to the so-called \emph{component Lagrangian} by means of performing a \emph{Berezin integral}. This method of investigation is known as the \emph{component formalism}, and in it one can employ techniques
 of differential geometry to derive various properties defining the physical system, such as the equations of motion. There are however obstructions to reducing the supermanifold Lagrangian to the component Lagrangian, one of the being the supermanifold itself.\footnote{see for instance \cite{WD} where this problem arises in the context of superstring theory} Nevertheless, both the supermanifold Lagrangian and the component Lagrangian can be defined as zero-cochains, and relations between them can be explored.

\begin{REM}\emph{
It should be noted that in \cite{FREEDSUSY} the component Lagrangian differs from the Berezin integral of the supermanifold Lagrangian by a Poincar\'e-invariant differential operator. However, in this paper, we are only concerned with the superparticle admitting two supersymmetries (i.e., where the supermanifold $\Xfr$ has odd dimension \emph{two}) and, as remarked in \cite[p. 77]{FREEDSUSY}, examples where the Poincar\'e-invariant differential operator is non-zero appear in theories with at least four supersymmetries. Hence, without loss of generality, we take the component Lagrangian here to be \emph{defined} as the Berezin integral of the supermanifold Lagrangian. 
}
\end{REM}

The component Lagrangian is described here in much the same way that the superparticle Lagrangian was described in Section \ref{jckrnckrnckjrkxe}. Let $\Phi\in \Fc$ denote a field and write  
\begin{align}
\Phi(Y, \eta^1, \eta^2) = \phi(y) + \psi_a(y)\eta^a + F(y)~ \eta^1\wedge \eta^2
\label{fjnckrnckrnckr}
\end{align}
where $a = 1, 2$ is being summed. The component fields of $\Phi$ is the four-tuple of fields $(\phi, \psi_a, F)$. However, by abuse of notation we will identify $\Phi$ with its components.\footnote{
If the supermanifold is split then $\Phi$ uniquely determines, and is uniquely determined, by its components. As our applications are all on split supermanifolds, there is no ambiguity in identifying $\Phi$ with its components.
} Regarding the Berezin integral, it may be thought of as a map from $\Ber~\Xfr$ to $\det M$, where $\det M =  \wedge^{\mathrm{top}}(T^*M)$ is the determinant line bundle of $M$. As $M$ is one-dimensional we may identify $\det M$ with the one-forms $\Om^1(M)$. Over $\Vc$ the Berezin integration is given by,
\[
\int_{\Ber}\eta^1\wedge \eta^2[\drm Y~\drm\eta^1\drm \eta^2] =\drm y
\]
and zero otherwise. To illustrate, we see from \eqref{fjnckrnckrnckr} that $\int_\Ber\Phi~[\drm Y~\drm\eta^1\drm \eta^2] =F~\drm y$. 
\\

Now let $\Ufr = \{\Uc, \Vc, \ldots\}$ denote a covering for $\Xfr_{(M, E)}$ and let $\Ufr_\red= \{U, V, \ldots\}$ denote a covering of $M$. Recall that the superparticle Lagrangian $\Lc(\Phi) = \{\Lc_\Uc(\Phi), \Lc_\Vc(\Phi), \ldots\}$ for $\Xfr$ was introduced as a zero-cochain with respect to $\Ufr$. Set,
\[
L_U(\Phi) := \int_{\Ber} \Lc_\Uc(\Phi).
\]
Thus we have a quantity defined on each open set $U\subset M$ and therefore we obtain a zero-cochain $L(\Phi):= \{L_U(\Phi), L_V(\Phi), \ldots\}$ valued in the sheaf of smooth sections of $\det M$, denoted $\underline \det^\8 M$. Now note that we have here two cochain complexes with which to wrestle: $(C^\bt(\Ufr, \underline\Ber^\8\Xfr), \dt)$ and $(C^\bt(\Ufr_\red, \underline\det^\8 M), \dt_\red)$. We remark that it need \emph{not} in general hold that the Berezin integration map will define a morphism of these differential complexes. In particular, there is no need for $\int_\Ber (\dt\Lc)_{\Uc\Vc}(\Phi) = (\dt_\red L)_{UV}(\Phi)$ to hold for all $\Phi$. This leads us to the following definition.

\begin{DEF}
\label{rkclrcklrmclkr}
\emph{
A field $\Phi$ will be called a \emph{good field} for the Lagrangian $\Lc$ on $\Xfr$ if the Berezin integral commutes with the coboundary operator along $\Lc(\Phi)$, i.e., if
\begin{align*}
\int_\Ber (\dt\Lc)(\Phi) = (\dt_\red L)(\Phi). 
\end{align*}
}
\end{DEF}

We denote the set of \emph{good fields} for $\Lc$ by $\Gc_\Lc$. It is a subset of $\Fc$. A characterisation of $\Gc_\Lc$ is not undertaken in this paper and we defer a more detailed discourse on $\Gc_\Lc$ to a forthcoming work \cite{BETT}. Throughout the rest of this paper, the results derived will be valid only on the subset of good fields $\Gc_\Lc$ of $\Fc$, unless otherwise stated.

\subsection{Consistency of the Component Lagrangian}
As mentioned at the start of Section \ref{djckncjkrnckjr}, once one has the component Lagrangian, techniques of differential geometry become applicable. In this paper we are interested, in particular, in the variational derivative, whose extrema are taken to satisfy the equations of motion. The notion of \emph{consistency} that we formulate here is concerned with the ability to ``consistently'' define the variational derivative of the component Lagrangian. To elaborate, firstly fix a set of fields $\Fc$ for the Lagrangian $\Lc$ and denote by $\dt_\Fc$ the variational derivative. Recall that the component Lagrangian $L$ is a 0-cochain on $M$. In each open set $U\subset M$ it makes sense to take the variational derivative of $L_U(\Phi)$, denoted $\dt_\Fc L_U(\Phi)$.\footnote{
Recall that for a functional $L$, depending on a field $f(x)$, the functional, or variational derivative along $f$ is given by
\[
\frac{\dt L}{\dt f} = \left(\frac{\pt L}{\pt f} - \frac{\pt}{\pt x}\left( \frac{\pt L}{\pt f^\p}\right)\right)\dt f.
\]
} Hence it is natural to require the following on all non-empty intersections $U\cap V$:
\begin{align}
\mbox{$\dt_\Fc L_U(\Phi) = 0$ iff $\dt_\Fc L_V(\Phi) = 0$.}
\label{fjcnkrnckr}
\end{align}
The variations $\dt_\Fc L_U$ are termed \emph{infinitesimal} since they are defined locally. If \eqref{fjcnkrnckr} holds, we see that these variations are in fact \emph{global}, i.e., on $U\cap V$ that $\dt_\Fc L_U(\Phi)\propto \dt_\Fc L_V(\Phi)$ for all $\Phi$. This motivates the following definition. 

\begin{DEF}
\label{jbckncjkcnkr}
\emph{
The component Lagrangian $L$ is said to be \emph{consistent} if its infinitesimal variations are global. 
}
\end{DEF}

From a physical perspective, it would be desirable for the given component Lagrangian to be consistent. This is because consistency here represents the statement that: \emph{the equations of motion are independent of choice of coordinate system on $M$}. Now, in general, the component Lagrangian will not be consistent. However, we can force it to be consistent by restricting the set of fields $\Fc$ that we consider in the field theory. This leads to the notion of a \emph{consistent field} as follows. 

\begin{DEF}
\emph{
A field $\Phi$ is said to be \emph{consistent} for the component Lagrangian $L$ if $L$ is consistent along $\Phi$.
}
\end{DEF}

The locus of consistent fields for $L$ will be denoted $\Fc_L$. It is a subset of $\Fc$. A \emph{good} consistent field for $L$ is a consistent field which is also good, i.e., it resides in the intersection $\Fc_L\cap \Gc_\Lc$, which may or may not be empty. In what remains of this paper we intend to characterise this locus of good, consistent fields (presuming the set of such fields is non-empty). In the case where $\Lc$ is global we have the following result. 

\begin{PROP}\label{djncjklnrjcmrjcmlr}
Let $\Xfr$ be $(1|2)$-dimensional and maximally superconformal and suppose $\Lc$ is a global superparticle Lagrangian on $\Xfr$. Then the set of good, consistent fields, if non-empty, is given by
 \begin{align}
\label{fjckcnjknrk}
\Fc_L\cap \Gc_\Lc = \left\{ \Phi\in \Gc_\Lc \mid \dt_\Fc\left(\e^*\Lc(\Phi)\right) = 0\right\}.
\end{align}
Moreover, if the cocycle $h$ defining the superconformal structure satisfies $\e^*h = h$, then we have $\Fc_L\cap \Gc_\Lc = \Gc_\Lc$. 
\end{PROP}

\begin{proof}
Let $\Lc$ denote the superparticle Lagrangian for $\Xfr$ and $L$ the component Lagrangian and let $(\Ufr, \rho)$ denote a trivialisation of $\Xfr$. If $\Lc$ is global, then recall from \eqref{jcnkrnckjrcnklrcmlr} that we can write $(\rho_{\Uc\Vc})_*\Lc_\Uc = h_{\Uc\Vc} \Lc_\Vc$, for $h =\{h_{\Uc\Vc}\}$ the defining cocycle for $\Lc$. Set $h_{\Uc\Vc} = h_{UV, 0} + h_{UV, 12}~\eta^1\wedge\eta^2$ in $\Vc$. In performing a Berezin integral we find,
\begin{align}
\int_\Ber(\rho_{\Uc\Vc})_*\Lc_\Uc(\Phi) - h_{\Uc\Vc,0} L_V(\Phi) = h_{UV, 12}~\e^*(\ell_\Vc(\Phi)),
\label{djcnkjncjkrnckjrnk}
\end{align}
where $\ell_\Vc$ is the functional component of $\Lc_\Vc$, i.e., that $\Lc_\Vc(\Phi) = \ell_\Vc(\Phi)[\drm Y\drm\eta^1\drm\eta^2]$. Now if $\Phi$ is a good field, it will follow that 
\begin{align}
f_{VU}^* L_U(\Phi) = \int_\Ber(\rho_{\Uc\Vc})_*\Lc_\Uc(\Phi)
\label{djncknfcknrcknr}
\end{align}
where $f = \rho_\red$ are the transition functions for the underlying manifold $M$ of $\Xfr_{(M, E)}$. Indeed \eqref{djncknfcknrcknr} is the motivation behind Definition \ref{jbckncjkcnkr}. Now from \eqref{djcnkjncjkrnckjrnk} we have the following expression on the intersection $U\cap V$,
\begin{align}
L_U(\Phi) - h_{\Uc\Vc,0} L_V(\Phi) = h_{UV, 12}~\e^*(\Lc_\Vc(\Phi)).
\label{fkmrlmklrmvlr}
\end{align}
Applying the variational derivative $\dt_\Fc$ to the above expression yields,
\begin{align}
\dt_\Fc L_U (\Phi)- h_{UV, 0}~ \dt_\Fc L_V(\Phi) = h_{UV, 12}~ \dt_\Fc (\e^*\Lc_V(\Phi)).
\label{rverrverv}
\end{align}
Hence the first part of this proposition \eqref{fjckcnjknrk} follows. As for the second, claiming $\Fc_L\cap \Gc_\Lc = \Gc_\Lc$, note that this follows from \eqref{fkmrlmklrmvlr} since $\e^*h = h$ is equivalent to $\{h_{UV, 12}\} = 0$.
\end{proof}

\section{The Superparticle on $\Cbb\Pbb^1$}
\label{rjnvkrnkrmjkvr}

In this section we consider the superparticle Lagrangian on a maximally superconformal supermanifold over $\Cbb\Pbb^1$. The goal here will be to explicitly describe the locus $\Fc_L\cap \Gc_\Lc$ of good, consistent fields for the component Lagrangian $L$. In the case where $\Lc$ is global, and $\Xfr$ is $(1|2)$-dimensional and maximally superconformal, we had the simple characterisation of these fields in Proposition \ref{djncjklnrjcmrjcmlr}. Unfortunately, this cannot be readily applied in this section. This is because, as we shall show, there cannot exist any global superparticle Lagrangian which is compatible with the superconformal structure.\footnote{
This of course does \emph{not} imply $\Lc$ cannot be globally defined. Only that, should it be globally defined, it \emph{cannot} be compatible with the superconformal structure. 
}

\subsection{A Non-Existence Result}
To begin, firstly recall Corollary \ref{rjcnkrnjckrncjkr} which asserts (essentially) that the set of $(1|2)$-dimensional, maximally superconformal, supermanifolds over $\Cbb\Pbb^1$ may be identified with the set of rank-two vector bundles of degree-$(-2)$. Let $\Pi_{(-k_1, -k_2)}$ be such a supermanifold over $\Cbb\Pbb^1$. With respect to the standard trivialisation of $\Cbb\Pbb^1$, the defining cocycle for the superconformal structure is
\begin{align}
g_{\Uc\Vc, a}(Y, \eta^1, \eta^2) = \lam_ay^{k_a} - (-1)^{a-1}k_a \lam_a y^{k_a - 1}\eta^1\wedge\eta^2,
\label{jcnkrckrmckrmckr}
\end{align}
for $a = 1, 2$. This was obtained from the more general formula, derived in \eqref{rvjnkrnckrckjr}, applied to the trivialisation given in \eqref{ejxnnxjkenxjken}. We now have the following negative result. 

\begin{PROP}
Let $\Xfr$ be a $(1|2)$-dimensional, maximally superconformal supermanifold over $\Cbb\Pbb^1$ with its standard trivialisation. A global superparticle Lagrangian, compatible with the superconformal structure on $\Xfr$, does not exist. 
\end{PROP}

\begin{proof}
To prove this result we appeal to Proposition \ref{ejcnknckrncjrknr}. It was identified there and in the subsequent discussion that a global superparticle Lagrangian, compatible with the superconformal structure, exists if and only if the hermitian product of the defining cocycles for the superconformal structure is purely real. Then, to prove this proposition, we need only show here that this is not the case.   

Firstly, we know that any maximally superconformal supermanifold over $\Cbb\Pbb^1$ must be of the form $\Pi_{(-k_1, -k_2)}$, where $k_1 + k_2 = 2$. The defining cocycles for the superconformal structure on $\Pi_{(-k_1, -k_2)}$ is explicitly given in \eqref{jcnkrckrmckrmckr}. Let $G_{\Uc\Vc}$ denote the hermitian product of $g_{\Uc\Vc, 1}$ and $g_{\Uc\Vc, 2}$. Then we see that
\begin{align}
G_{\Uc\Vc} &= g_{\Uc\Vc, 1}\overline{g_{\Uc\Vc, 2}}\notag\\
&= \lam_1\overline{\lam_2} ~y^{k_1}\overline y^{k_2} \left( 1 + \left( k_2\overline y^{-1} -  k_1y^{-1}\right)\eta^1\wedge\eta^2\right).
\label{jcnkncknkrnkrcrrc}
\end{align}
Therefore, in order for $G_{\Uc\Vc}$ to be purely real, we necessarily require $k_1 = k_2$. Since $k_1 + k_2 = 2$, it follows that $k_1 = k_2 = 1$. But then, in writing $y = e^{i\vp}$ for a real variable $\vp$, we see that 
\begin{align}
\label{dcjnknckjfncjnfk}
G_{\Uc\Vc} = \lam_1\overline{\lam_2} (1 - 2i\sin\vp~\eta^1\wedge\eta^2).
\end{align}
Thus the imaginary part of $G_{\Uc\Vc}$ does not vanish identically, and as such obstructs the existence of a global superparticle Lagrangian compatible with the superconformal structure. 
\end{proof}

In spite of the above result we may nevertheless make sense of the notion of consistency for the component Lagrangian and characterise its locus of consistent fields.

\subsection{Consistent fields on $\Cbb\Pbb^{1|2}$} Consider the superparticle sigma-model on $\Cbb\Pbb^{1|2} = \Pi_{(-1, -1)}$ given by the superparticle Lagrangian $\Lc$ and the set of fields $\Fc$ as in \eqref{ecfcjhbhcjbdfc}. Recall that any field $\Phi$ may be written,
\begin{align}
\Phi = \phi + \psi_a\eta^a + F~\eta^1\wedge\eta^2.
\label{dkmclncrucrxe}
\end{align}
Then we have,
\begin{align}
\label{djnxkjnjkenenxs}
D_{\Vc, 1}\Phi = \psi_1 + (F - \phi^\p)\eta^2 + \psi^\p_1~\eta^1\wedge\eta^2
&&
\mbox{and}
&&
D_{\Vc, 2}\Phi = \psi_2 - (F +\phi^\p)\eta^1 - \psi_2^\p~\eta^1\wedge\eta^2.
\end{align}
In this way $\e^*D_a\Phi = \psi_a$ may be identified with a \emph{smooth} section of $\Oc(-1)$. We now have the following characterisation of the set of consistent fields. 

\begin{PROP}
\label{jnckncjknrkcnrkxe}
The locus $\Fc_L\cap \Gc_\Lc$ of good, consistent fields for $L$ is given by 
\[
\Fc_L \cap\Gc_\Lc \cong \left\{
\Phi\in \Gc_\Lc\mid \mbox{$(\e^*D_1\Phi, \e^*D_2\Phi)$ are real and imaginary (resp.), or vice-versa}
\right\}.
\]
\end{PROP}

\begin{proof}
Let $\Lc$ be the superparticle Lagrangian for $\Pi_{(-1, -1)} = \Cbb\Pbb^{1|2}$. From \eqref{jcnkncjknckjncke} and \eqref{dcjnknckjfncjnfk} we see, on the intersection $\Uc\cap\Vc$, that $\Lc_\Uc$ and $\Lc_\Vc$ differ as follows
\begin{align}
(\rho_{\Uc\Vc})_*\Lc_\Uc - \Lc_\Vc = -2i\sin\vp~\eta^{12}\Mcl_\Vc. 
\label{cjnkencknckenck}
\end{align}
We have set $\lam_1 = \lam_2 = i$ without loss of generality. Along good fields $\Phi$, the Berezin integral applied to \eqref{cjnkencknckenck} gives the corresponding difference of $L_U(\Phi)$ and $L_V(\Phi)$ on $U\cap V$,
\begin{align}
L_U(\Phi) - L_V(\Phi) &= - 2i\sin\vp \int_\Ber\eta^{12}\Mcl_\Vc(\Phi)\notag \\
&= - 2i\sin\vp~(\e^*m_\Vc(\Phi))~\drm y,
\label{jekjncjkenckjeckec}
\end{align}
where $\Mcl_\Vc(\Phi) = m_\Vc(\Phi)~[\drm Y~\drm\eta^1\drm\eta^2]$ and, from \eqref{jcnkrnckrmckjmrkjr}, 
\begin{align}
m_\Vc(\Phi) &= \frac{1}{2}s^{ab}\left\langle D_{\Vc,a}(\Phi), D_{\Vc, b}(\Phi)\right\rangle
\notag 
\\
&= \frac{1}{2} \left( \left\langle D_{\Vc,1}(\Phi), D_{\Vc, 2}(\Phi)\right\rangle + \left\langle D_{\Vc,2}(\Phi), D_{\Vc, 1}(\Phi)\right\rangle\right).
\label{fjnckcmkmcklrcr}
\end{align}
Now recall from Definition \ref{jbckncjkcnkr} that, for $L$ to be consistent, the infinitesimal variations $\dt L_U$ and $\dt L_V$ must be global. As such we see from \eqref{jekjncjkenckjeckec} that the component Lagrangian will be consistent along $\Phi$ iff $\dt_\Fc(\e^*m_\Vc(\Phi)) = 0$. We have,
\[
\e^*m_\Vc(\Phi)  = \psi_1\overline{\psi_2} + \overline{\psi_1}\psi_2
\]
which yields the following,
\[
\dt_\Fc (\e^*m_\Vc(\Phi)) = \left( \psi_1~\dt_\Fc\overline{\psi_2} + \overline{\psi_1}~\dt_\Fc\psi_2\right) + 
 \left( \psi_2~\dt_\Fc\overline{\psi_1} + \overline{\psi_2}~\dt_\Fc\psi_1\right).
\]
The above variation now vanishes if and only if $\psi_1$ and $\psi_2$ are respectively real and imaginary, or vice-versa. 
\end{proof}

In the background of a non-standard $\Cbb\Pbb^{1|2}$, i.e., some $\Pi_{(-k_1, -k_2)}$, for $(k_1, k_2)\neq (1, 1)$, we can no longer appeal to the description of $G_{\Uc\Vc}$ in \eqref{dcjnknckjfncjnfk}, and must instead resort to the more complicated expression in \eqref{jcnkncknkrnkrcrrc}. Interestingly, this results in quite a different characterisation of the consistent fields than that obtained in Proposition \ref{jnckncjknrkcnrkxe}. 

\subsection{Consistent fields on a non-Standard $\Cbb\Pbb^{1|2}$}
We start off with the following lemma, elaborating on the calculation of the hermitian product $G_{\Uc\Vc}$ of the cocycles defining the superconformal structure.

\begin{LEM}\label{jfcnkkjnjkrvrles}
Let $\Re_{\Uc\Vc}$ and $\Im_{\Uc\Vc}$ denote the real and imaginary parts, respectively, of the hermitian product $G_{\Uc\Vc}$ of the defining cocycles for the superconformal structure of $\Pi_{(-k_1, -k_2)}(\lam)$. Suppose now that $k_1\neq k_2$. Then we can write,
\begin{align*}
\Re_{\Uc\Vc}(Y, \eta) = R_{UV, 0}(y) + R_{UV, 12}(y)~\eta^{12}
&&
\mbox{and}
&&
\Im_{UV} (Y, \eta) = I_{UV, 0}(y) + I_{UV, 12}(y)~\eta^{12}
\end{align*}
for non-vanishing, smooth functions $R_{UV, 0}, R_{UV, 12}, I_{UV, 0}$ and $I_{UV, 12}$ on $U\cap V$. 
\end{LEM}

\begin{proof}
Recall the expression for $G_{\Uc\Vc}$ from \eqref{jcnkncknkrnkrcrrc}. For simplicity, write $y = e^{i\vp}$ for a real variable $\vp$. Without loss of generality we take $\lam_1 = \lam_2 = i$. Since $\Pi_{(-k_1, -k_2)}$ admits a maximally superconformal structure, we have $k_1 + k_2 = 2$. Using this we see that \eqref{jcnkncknkrnkrcrrc} becomes,
\begin{align}
G_{\Uc\Vc} = e^{2i(k_1 - 1)\vp} (1 - 2k_1\cos \vp~\eta^{12}) - 2e^{i(2k_1-3)\vp}\eta^{12}. 
\label{cjnknckrnckjr}
\end{align}
As a consistency check note that we recover \eqref{dcjnknckjfncjnfk} for $k_1 = 1$. Now set $\tilde k_1 := k_1 - 1$. Then from \eqref{cjnknckrnckjr} and various trigonometric identities, the real and imaginary parts are given by
\begin{align}
\Re_{\Uc\Vc} &= \cos (2\tilde k_1\vp) - 2\left( \tilde k_1(\cos\vp)(\cos (2\tilde k_1\vp)) - (\sin \vp)(\sin(2\tilde k_1\vp))\right)\eta^{12}
\label{djcnkncjkrncrnrn}\\
\Im_{\Uc\Vc} &= \sin(2\tilde k_1\vp) - 2\left(
\tilde k_1(\cos\vp)(\sin(2\tilde k_1\vp)) + (\sin \vp)(\cos(2\tilde k_1\vp))\right) \eta^{12}.
\label{dkcmmclkmcklrcr}
\end{align}
As $k_1\neq1$, we see that $\tilde k_1\neq0$.  This completes the proof.
\end{proof}

As a result of Lemma \ref{jfcnkkjnjkrvrles} and \eqref{jcnkncjknckjncke} we see that the difference of the component Lagrangians on $U\cap V$ (along good fields) is,
\begin{align}
L_U - R_{UV,0}L_V =  \left( R_{UV, 12}~\e^*\ell_\Vc + iI_{UV, 12}~\e^*m_\Vc + iI_{UV,0} ~m_V\right)\drm y
\label{kcmlkfmcklrmclk}
\end{align}
where $\Lc_\Vc = \ell_\Vc[\drm y~\drm\eta^1\drm\eta^2]$ and $m_V\drm y = \int_\Ber \Mcl_\Vc$. Hence, in order for $L$ to be consistent along $\Phi$ we need the variation on the right-hand side of \eqref{kcmlkfmcklrmclk} to vanish along $\Phi$. Firstly, from \eqref{djnxkjnjkenenxs}, we have
\begin{align*}
\e^*\ell_\Vc &= \psi_1\overline{\psi_2} - \overline{\psi_1}\psi_2\\
\e^*m_\Vc &= \psi_1\overline{\psi_2} + \overline{\psi_1}\psi_2\\
M_V &=\frac{1}{2} \left( \overline{\psi_2}\psi_1^\p - \overline{\psi_1}\psi_2^\p + \psi_2\overline{\psi^\p_1} - \psi_1\overline{\psi_2^\p}\right)
 + \left( F\overline{\phi}^\p - \phi^\p\overline F\right).
\end{align*}
The variations of the above quantities are then
\begin{align*}
\dt_\Fc(\e^*\ell_\Vc) &= \left( \overline{\psi_2}~\dt_\Fc\psi_1 -\psi_2~\dt_\Fc\overline{\psi_1} \right) - \left( \overline{\psi_1}~\dt_\Fc\psi_2 -\psi_1~\dt_\Fc \overline{\psi_2}\right)
\\
\dt_\Fc(\e^*m_\Vc) &= \left(\overline{\psi_2}~\dt_\Fc\psi_1 +  \psi_2~\dt_\Fc\overline{\psi_1} \right) + \left(  \overline{\psi_1}~\dt_\Fc\psi_2 + \psi_1~\dt_\Fc\overline{\psi_2} \right) 
\\
\dt_\Fc M_V &= - \left( \overline{\psi_2}^\p~\dt_\Fc \psi_1 + \psi_2^\p~\dt_\Fc\overline{\psi_1}\right) + 
\left( \overline{\psi_1}^\p~\dt_\Fc \psi_2 + \psi_1^\p~\dt_\Fc \overline{\psi_1}\right) +  \dt_\Fc \left( F\overline{\phi}^\p - \phi^\p\overline F\right).
\end{align*}
Clearly, we see that a \emph{necessary} condition for $L$ to be consistent is the vanishing of $ \dt_\Fc ( F\overline{\phi}^\p - \phi^\p\overline F)$. Computing this we have,
\[
\dt_\Fc \left( F\overline{\phi}^\p - \phi^\p\overline F\right) 
= 
\left( \overline\phi^\p \dt_\Fc F - \phi^\p \dt_\Fc \overline F\right)
+
\left( 
\overline F^\p\dt_\Fc \phi - F^\p\dt_\Fc\overline \phi \right).
\]
The above variation now vanishes iff both $\phi^\p$ and $F$ are either purely real or imaginary. Now write $F = \e^*D_1D_2\Phi$ and $\phi^\p = \e^*D^2\Phi$. These quantities may be identified with smooth sections of the tangent bundle. We then have the following.

\begin{PROP}\label{dicdcnnckelcekm}
Suppose $\Phi$ is a good field for $\Lc$ and $L$ is consistent along $\Phi$. Then both $\e^*D^2\Phi$ and $\e^*D_1D_2\Phi$ are either purely real or imaginary. \qed
\end{PROP}

Another way to phrase the content in Proposition \ref{dicdcnnckelcekm} is to write,
\begin{align}
\Fc_L\cap \Gc_\Lc
\subset
\left\{
\Phi\in\Gc_\Lc \mid  \mbox{$(\e^*D^2\Phi,\e^*D_1D_2\Phi)$ are both either purely real or imaginary}
\right\}.
\label{fjcnkfncjkfncjkfnck}
\end{align}
Hence the conditions on $\Phi$ in Proposition \ref{dicdcnnckelcekm} are necessary, but need not be sufficient. This is in contrast with Proposition \ref{jnckncjknrkcnrkxe} where the locus $\Fc_L\cap \Gc_\Lc$ was fully characterised. We consider in more detail the case where $\tilde k_1 = 1$, which is the supermanifold $\Pi_{(-2, 0)}$. Our goal is to give a more complete description of $\Fc_L\cap \Gc_\Lc$ than that obtained in \eqref{fjcnkfncjkfncjkfnck}.

\subsubsection{The Supermanifold $\Pi_{(-2, 0)}$} Note firstly that \eqref{djcnkncjkrncrnrn} and \eqref{dkcmmclkmcklrcr} become
\begin{align*}
\Re_{\Uc\Vc} = \cos 2\vp - 2\cos 3\vp~\eta^{12} 
&&
\mbox{and}
&&
\Im_{\Uc\Vc} = \sin 2\vp - 2\sin 3\vp~\eta^{12}. 
\end{align*}
Assuming the necessary conditions in Proposition \ref{dicdcnnckelcekm}, the variation on the right-hand side of \eqref{kcmlkfmcklrmclk} becomes
\begin{align}
\label{djcnkncjkrnckjrnk}
\dt_\Fc (L_U - R_{UV,0}L_V)=&~ - \left( 2e^{3i\vp}\overline{\psi_2} + i\sin 2\vp~\overline{\psi_2}^\p\right)\dt_\Fc\psi_1+ \left(2e^{-3i\vp}\psi_2 - i\sin 2\vp~\psi_2^\p\right)\dt_\Fc\overline{\psi_1}\\ \notag
&~+ \left( 2e^{-3i\vp}\overline{\psi_1} + i\sin 2\vp~\overline{\psi_1}^\p\right)\dt_\Fc\psi_2 - \left( 2e^{3i\vp}\psi_1 - i\sin 2\vp~\psi_1^\p\right)\dt_\Fc\overline{\psi_2}.
\end{align}
We are interested in identifying necessary and sufficient conditions for the above variation to vanish, thereby ensuring consistency of $L$ along $\Phi$. To that extent, define the following first-order differential operator,
\[
\Delta_{(m, n)}(\vp) := 2 e^{-mi\vp} + e^{-i\vp}\sin n\vp~\frac{\pt}{\pt \vp}.
\]
We then have the following:

\begin{PROP}
\label{djcnknckecnjkekce}
Suppose $\Phi$ is such that it satisfies the necessary conditions in Proposition \ref{dicdcnnckelcekm} and
\begin{align*}
 \overline{\Delta_{(3, 2)}(\vp)}\left( \e^*D_1\Phi\right) = 0  && \mbox{and} &&\overline{\Delta_{(-3, 2)}(\vp)}\left( \e^*D_2\Phi\right) = 0.
\end{align*}
Then the component Lagrangian $L$ will be consistent along $\Phi$. 
\end{PROP}

\begin{proof}
Note from \eqref{djcnkncjkrnckjrnk} that we can write
\begin{align*}
\dt_\Fc (L_U - R_{UV,0}L_V) =&~-  \Delta_{(-3, 2)}(\vp) \overline{\psi_2}~\dt_\Fc\psi_1 
+ 
\overline{\Delta_{(-3, 2)}(\vp)}\psi_1~ \dt_\Fc\overline{\psi_1}\\
&~+
 \Delta_{(3, 2)}(\vp)\overline{\psi_2}~\dt_\Fc\psi_2 - \overline{\Delta_{(3, 2)}(\vp)}\psi_2~\dt_\Fc\overline{\psi_2}\\
 =&~- \overline{\left(\overline{\Delta_{(-3, 2)}(\vp) }\psi_2\right)}\dt_\Fc\psi_1 
+ 
\left(\overline{\Delta_{(-3, 2)}(\vp)}\psi_1\right) \dt_\Fc\overline{\psi_1}\\
&~+
\overline{\left( \overline{\Delta_{(3, 2)}(\vp)}\psi_2\right)}~\dt_\Fc\psi_2 - \left(\overline{\Delta_{(3, 2)}(\vp)}\psi_2\right)\dt_\Fc\overline{\psi_2}.
\end{align*}
The proposition now follows. 
\end{proof}

Hence Proposition \ref{djcnknckecnjkekce} above, coupled with Proposition \ref{dicdcnnckelcekm}, describe necessary and sufficient conditions for $L$ to be consistent along $\Phi$, thereby providing a complete characterisation of the locus $\Fc_L\cap\Gc_\Lc$ for $\Pi_{(-2, 0)}$. Under further mild assumptions on $\Phi$ we can obtain another characterisation of this locus as follows. Suppose firstly that $\Phi$ satisfies the conditions in Proposition \ref{dicdcnnckelcekm} and moreover that $\e^*D_1\Phi$ and $\e^*D_2\Phi$ are both either purely real or purely imaginary. Then we say $\Phi$ has type $(\Re, \Re)$ if both $\e^*D_1\Phi$ and $\e^*D_2\Phi$ are purely real; and type $(\Re, \Im)$ if $\e^*D_1\Phi$ and $\e^*D_2\Phi$ are purely real and imaginary respectively, and so on. In this notation we have the following.

\begin{PROP}
Suppose $\Phi$ is of type $(\Re, \Re)$ or $(\Im, \Re)$. Then $L$ if consistent if and only if $\e^*D_1\Phi = 0$. Similarly, if $\Phi$ is of the other types, i.e., $(\Re, \Im)$ or $(\Im, \Im)$, then $L$ is consistent if and only if $\e^*D_2\Phi = 0$.
\end{PROP}

\begin{proof}
We will prove this result in the case where $\Phi$ is of type $(\Re, \Re)$ or $(\Im, \Re)$, for when $\Phi$ is of the other types the proof is similar. So suppose firstly that $\Phi$ is of type $(\Re, \Re)$. Then the variation in \eqref{djcnkncjkrnckjrnk} becomes,
\begin{align}
\dt_\Fc (L_U - R_{UV,0}L_V)=  -4\sinh(3i\vp)\psi_2~\dt_\Fc\psi_1 - \left( 4\sinh(3i\vp)\psi_1 - 2i\sin (2\vp)\psi_1^\p\right)\dt_\Fc\psi_2.
\label{dcjnkdnckrncrcrd}
\end{align}
For $L$ to be consistent we require the right-hand side of \eqref{dcjnkdnckrncrcrd} to vanish. For this to happen, note that it is necessary for $\psi_2 \equiv 0$. But if $\psi_2$ vanishes identically, then so does $\dt_\Fc\psi_2$. Hence we see that $L$ will be consistent along $\Phi$ iff $\e^*D_1\Phi= 0$. 

In the case where $\Phi$ is of type $(\Im, \Re)$ the variation becomes,
\[
\dt_\Fc (L_U - R_{UV,0}L_V)=  -4\cosh(3i\vp)\psi_2~\dt_\Fc\psi_1 - \left( 4\cosh(3i\vp)\psi_1 + 2i\sin (2\vp)\psi_1^\p\right)\dt_\Fc\psi_2.
\]
We will reach the same conclusion by arguing just as in the case where $\Phi$ was of $(\Re, \Re)$-type. 
\end{proof}

\section{Concluding Remarks}

Upon contemplation on the objects arising in this paper, there are several natural questions to consider. Firstly, it should be observed that the supermanifolds appearing here are all of a particularly simple kind---that is, they are all \emph{split}. However, one of the key objects of interest, the superparticle Lagrangian $\Lc$, can be formulated on a supermanifold of any type---in particular on those of non-split type. Moreover, the consistency of the component Lagrangian $L$ did not \emph{a priori} require $\Lc$ to be global. As such, questions of consistency of $L$ may still be formulated on supermanifolds of non-split type. In doing so, one may perhaps gain insight into the nature of a non-split supermanifold itself. 
\\

One of the main objectives in this paper is to explore the relationship between the supersymmetric field theory defined on a supermanifold and the geometry of the supermanifold. For instance one could ask: \emph{which supermanifolds are best suited to the study of the field theory in question?} Should one be equipped with a particular supersymmetric Lagrangian, one can then attempt to answer this question by invoking methods introduced in this paper. In the case where the Lagrangian is that of the superparticle, we would answer that it is supermanifolds which admit, what we termed, a \emph{maximally superconformal structure}. It then becomes meaningful to study and classify those supermanifolds which admit such structures. Indeed, it is for this reason that we constrained our efforts in this paper to split supermanifolds, since it was shown that non-split supermanifolds could not admit maximal superconformal structures---at least in the $(1|2)$-dimensional case.
\\

We conclude with a remark here a possible generalisation. The superparticle discussed in this paper is considered with two supersymmetries---reflected by the odd-dimension of the supermanifold in question being \emph{two}. However, one need not be restricted to just two supersymmetries. Indeed, the local description of the superparticle Lagrangian $\Lc$ lends itself to a very logical generalisation to a superparticle Lagrangian with more supersymmetries, which in turn motivates a notion of superconformal structure in more generality. Given the stringent nature of the superconformal structure in dimension-$(1|2)$, a question of interest then becomes: \emph{how stringent is the superconformal structure in general?} For instance, supposing the superconformal structure has been generalised appropriately to $(p|q)$-dimensional supermanifolds, is there an analogue of Theorem \ref{chncnrcnrk}? Such a result may make the study of the superparticle with more-than-two supersymmetries on non-trivial supermanifolds quite tractable.

\appendix

\numberwithin{equation}{section}

\section{Proof of Theorem \ref{chncnrcnrk}}\label{ejnckenckrckrrc}

Let $\Ufr = \{\Uc, \Vc, \ldots\}$ denote an open covering of $\Xfr_{(M, E)}$ and let $\Ufr_\red = \{U, V, \ldots\}$ denote an open covering of $M$. Let $\rho = (\rho_{\Uc\Vc})$ be the transition functions of $\Xfr$ with respect to $\Ufr$ and denote by $\rho^+$ and $\rho^-$ the even and odd components. Explicitly, for coordinates $(X, \q^1,\q^2)$ on $\Uc$ and $(Y,\eta^1,\eta^2)$ on $\Vc$, we have
\begin{align}
Y \sim \rho^+_{\Uc\Vc}(X, \q) = f_{UV}(x) + \al_{UV}(x)~\q^1\wedge\q^2
&&
\mbox{and}
&&
\eta^a \sim \rho_{\Uc\Vc}^{-, a}(X, \q) = \zeta^a_{UV, b}(x)\q^b,
\label{cejhbjebchjeve}
\end{align}
where $a, b = 1, 2$; $\{f_{UV}\}$ denotes the transition functions of $M$; $\al_{UV}$ is a holomorphic function on $U\cap V$; and $\{\zeta_{UV}\} = \{(\zeta_{UV, b}^a)\}$ are the transition functions of the vector bundle $E$. 
\\

Suppose firstly that $\Xfr_{(M, E)}$ is maximally superconformal and let $g = \{(g^a_{\Uc\Vc,b})\}$ denote the defining cocycle for the superconformal structure. Then \eqref{jcnknckncjkck} holds for $(g^a_{\Uc\Vc,b})$ diagonal. The push-forward map on the superconformal vector fields may be computed explicitly using the transition functions in \eqref{cejhbjebchjeve}. On $D_{\Uc, 1}$ we have,
\begin{align}
(\rho_{\Uc\Vc})_*D_{\Uc, 1} = \left( \zeta_{UV, 1}^b - \frac{\pt \zeta_{UV, a}^b}{\pt x}\zeta_{VU, c}^2\zeta_{VU, d}^a~\eta^{cd}
\right)\frac{\pt}{\pt \eta^b}
-
\left(
\frac{\pt f_{UV}}{\pt x} - \al_{UV}
\right)\zeta_{VU, c}^2
\eta^c\frac{\pt}{\pt y}
\label{jcnkncjkrcnkrn}
\end{align}
where we have set $\eta^{cd} \equiv \eta^c\wedge\eta^d$ and the Einstein summation convention is being used.  The expression for $D_{\Uc, 2}$ is similar but contains some changes of sign. We find,
\begin{align}
(\rho_{\Uc\Vc})_*D_{\Uc, 2} =
\left( \zeta_{UV, 2}^b - \frac{\pt \zeta_{UV, a}^b}{\pt x}\zeta_{VU, c}^1\zeta_{VU, d}^a~\eta^{cd}
\right)\frac{\pt}{\pt \eta^b}
-
\left(
\frac{\pt f_{UV}}{\pt x} + \al_{UV}
\right)\zeta_{VU, c}^1
\eta^c\frac{\pt}{\pt y}.
\label{rjncknkcrnece}
\end{align}
Now since $\{g_{\Uc\Vc}\}$ is diagonal, it follows from \eqref{jcnknckncjkck} that $\zeta_{UV, 1}^2$ and $\zeta_{UV, 2}^1$ vanish. Hence $E$ must split into a sum of line bundles $\ell_1 \oplus \ell_2$, defined by transition functions $\zeta_{UV, 1}^1$ and $\zeta_{UV, 2}^2$.  Now set $g_{\Uc\Vc, a} = (g_{\Uc\Vc, a}^a)$. We have,
\begin{align}
g_{\Uc\Vc, a} = \zeta_{UV, a}^a + (-1)^{a-1}\det\zeta_{VU}\frac{\pt\zeta_{UV, a}^a}{\pt x}~\eta^{12},
\label{rvjnkrnckrckjr}
\end{align}
where $\det \zeta_{VU} = \zeta_{VU, 1}^1\zeta_{VU, 2}^2$.

\begin{REM}
\emph{
At this stage it is not obvious that the right-hand side of the expression for $g_{\Uc\Vc}$ in \eqref{rvjnkrnckrckjr} will satisfy the cocycle condition. However, as we assumed that $\Xfr$ is maximally superconformal, then $\{g_{\Uc\Vc}\}$ is assumed to be a 1-cocycle \emph{a priori}. In the converse statement of this theorem we will explore the right-hand side of \eqref{rvjnkrnckrckjr} in more detail.
}
\end{REM}

We consider now the $\pt/\pt y$-component of \eqref{jcnkncjkrcnkrn} and \eqref{rjncknkcrnece}. By maximal superconformality,
\begin{align*}
\eta^2 g_{\Uc\Vc, 1} = \left( \frac{\pt f_{UV}}{\pt x} - \al_{UV}\right) \zeta_{VU, 2}\eta^2
&&
\mbox{and}&&
\eta^1 g_{\Uc\Vc, 2} =\left( \frac{\pt f_{UV}}{\pt x} + \al_{UV}\right) \zeta_{VU, 1}^1\eta^1.
\end{align*}
Then from \eqref{rvjnkrnckrckjr} it follows that one can solve for $\al_{UV}$ in two different ways as follows,
\begin{align}
\al_{UV} = \det \zeta_{UV} - \frac{\pt f_{UV}}{\pt x} = - \left( \det \zeta_{UV} -  \frac{\pt f_{UV}}{\pt x}\right). 
\label{jcnknckrcnkrnckr}
\end{align}
This shows $\al_{UV} = -\al_{UV}$ which can hold iff $\al_{UV} = 0$. Hence we have shown that $\Xfr_{(M, E)}$ is just the split model $\Pi E$. 

Finally, that $\det E = TM$ follows from \eqref{jcnknckrcnkrnckr} and the observation just made that $\al_{UV}=0$. 
\\

As to the converse implication, suppose $E = \ell_1\oplus \ell_2$, where $\ell_a$, $a = 1, 2$, is defined by $\{\zeta_{UV, a}^a\}$ so that $\{\zeta_{UV}\} := \{\zeta_{UV, 1}^1\oplus \zeta_{UV, 2}^2\}$ defines $E$. Now write $g_{\Uc\Vc, a}$ in terms of these transition functions as in \eqref{rvjnkrnckrckjr}. Then \emph{if} $\{g_{\Uc\Vc, a}\}$ satisfies the cocycle condition, it will follow that $\Pi E$ is maximally superconformal. Thus it remains to show that $\{g_{\Uc\Vc}\}$ satisfies the cocycle condition iff $\det E = TM$. This is now a simple computation. Recall that the cocycle condition requires,
\begin{align}
g_{\Uc\Vc, a}g_{\Vc\mathcal W, a} = g_{\Uc\Wc, a}.
\label{cjkncknrkcjrjkc}
\end{align}
We have here,
\begin{align*}
g_{\Uc\Vc, a} &= \zeta_{UV, a}^a + (-1)^{a-1}\det\zeta_{VU}\frac{\pt\zeta_{UV, a}^a}{\pt x}~\eta^{12}\\
g_{\Vc\Wc, a} &= \zeta_{VW, a}^a + (-1)^{a-1}\det\zeta_{WV}\frac{\pt\zeta_{VW, a}^a}{\pt y}~\gam^{12}\\
g_{\Uc\Wc, a} &= \zeta_{UW, a}^a + (-1)^{a-1}\det\zeta_{WU}\frac{\pt\zeta_{UW, a}^a}{\pt x}~\gam^{12}
\end{align*}
where $(\gam^1, \gam^2)$ denote the odd coordinates on $\Wc\in \Ufr$ and $\eta^{12} = \det \zeta_{WV}~\gam^{12}$. Imposing \eqref{cjkncknrkcjrjkc}, we find that the following must be satisfied
\begin{align}
\det \zeta_{WU}\frac{\pt\zeta_{UW, a}^a}{\pt x} &\stackrel{\mathrm{set}}{=}  \det\zeta_{VU}\det \zeta_{WV}~\frac{\pt\zeta_{UV, a}^a}{\pt x}\zeta_{VW, a}^a
+
 \det\zeta_{WV}~\zeta_{UV, a}^a\frac{\pt\zeta_{VW, a}^a}{\pt y} 
 \label{jnckncknrcrck}
  \\
 &= 
 \det\zeta_{WU}~\frac{\pt\zeta_{UV, a}^a}{\pt x}\zeta_{VW, a}^a
+
 \det\zeta_{WV}\frac{\pt f_{VU}}{\pt y}~\zeta_{UV, a}^a\frac{\pt\zeta_{VW, a}^a}{\pt x}.
 \label{fjcnkncjkrnckjr}
\end{align}
Now since $\zeta_{UW, a}^a = \zeta_{UV, a}^a \zeta_{VW, a}^a$, we see that the left-hand side of \eqref{jnckncknrcrck} becomes
\[
\det \zeta_{WU}\frac{\pt\zeta_{UW, a}^a}{\pt x} = \det \zeta_{WU}~\frac{\pt\zeta_{UV, a}^a}{\pt x} \zeta_{VW, a}^a+ \det \zeta_{WU}~\zeta_{UV, a}^a\frac{\pt\zeta_{VW, a}^a}{\pt x}.
\]
In comparing the above expression with \eqref{fjcnkncjkrnckjr} we see that \eqref{jnckncknrcrck} can hold iff 
\[
\det \zeta_{WU} =  \det\zeta_{WV}\frac{\pt f_{VU}}{\pt y} \iff  \frac{\pt f_{VU}}{\pt y}  = \det \zeta_{VU}.
\]
Thus we arrive at the condition $\det E = TM$. This completes the proof.

\hfill
\\
\noindent
\small
\textsc{
Kowshik Bettadapura, 
Mathematical Sciences Institute, Australian National University, Canberra, ACT 2601, Australia}
\\
\emph{E-mail address:} \href{mailto:kowshik.bettadapura@anu.edu.au}{kowshik.bettadapura@anu.edu.au}

\end{document}